\documentclass[12pt,a4paper]{amsart}
\newtheorem{theorem}{Theorem}[section]

\newtheorem{corollary}[theorem]{Corollary}
\newtheorem{proposition}[theorem]{Proposition}

\theoremstyle{definition}
\newtheorem{definition}[theorem]{Definition}

\newtheorem{example}[theorem]{Example}
\newtheorem{conjecture}[theorem]{Conjecture}

\newcommand{\norm}[1]{\Vert #1 \Vert}

\newcommand{\C}{\mathbb{C}}

\newcommand{\I}{\mathrm{i}}
\newcommand{\R}{\mathbb{R}}
\newcommand{\Z}{\mathbb{Z}}
\newcommand{\T}{\mathbb{T}}
\newcommand{\G}{\mathbb{G}}

\newcommand{\<}{\left<}
\renewcommand{\>}{\right>}

\newcommand{\econv}{\circledast}
\usepackage{amssymb}  
\usepackage{amsmath}
\usepackage{tikz}
\usetikzlibrary{matrix,arrows}

\begin{document}

\title[The inverse problem of pure point diffraction]{The inverse problem of pure point diffraction -- examples and open questions}

\author{Venta Terauds}

\address{School of Mathematics and Physics, University of Tasmania, GPO Box 252C-37, Hobart, Tasmania 7001, Australia}

\email{venta@hilbert.maths.utas.edu.au}

\thanks{Much of this research was undertaken at the University of Bielefeld with the support of the RCM$^2$ and the German Research Council (DFG), via the CRC 701. The author would like to thank Michael Baake for facilitating her stay at the University of Bielefeld and for many helpful 
discussions relating to this work.}

\begin{abstract}
This paper considers some open questions related to the inverse problem of pure point diffraction, in particular, what types of objects may diffract, and which of these may exhibit the same diffraction. Some diverse objects with the same simple lattice diffraction are constructed, including a tempered distribution that is not a measure, and it is shown that there are uncountably many such objects in the diffraction solution class of any pure point diffraction measure with an infinite supporting set.
\end{abstract}

\maketitle

\section{Introduction}

Whilst diffraction has been a central tool in the analysis of physical structures since early last century, the rigorous mathematical study of diffraction in terms of measure theory originated less than 20 years ago with the work of Hof \cite{hof95}. A structure may be represented mathematically as a measure $\omega$ and its diffraction is then also a measure, this defined to be the Fourier transform of the autocorrelation of $\omega$. When $\omega$ represents a regular physical structure such as a crystal or a quasicrystal, the diffraction may always be calculated, however in general, determining for exactly which structures the diffraction exists (that is, the autocorrelation may be calculated) is non-trivial. We refer the reader to \cite{bombtay} and \cite{bms} for partial answers to this question, and to \cite{cowley} and \cite{baagrimm11} for a comprehensive overview of physical and mathematical diffraction.

Once a diffraction is known to exist, consideration turns to the type of diffraction in evidence: does it exhibit only Bragg peaks (the diffraction measure is pure point), is it a diffuse scattering (the diffraction measure is continuous), or is it a combination of these?
On the other hand, one may begin with a diffraction, that is, a positive, centrally symmetric measure $\widehat{\gamma}$, and search for a structure or structures that have diffraction $\widehat{\gamma}$. This is the inverse diffraction problem. Structures which exhibit the same diffraction are usually termed homometric; here, as in \cite{teba}, we refer to the class of objects with a given diffraction $\widehat{\gamma}$ as the diffraction solution class of $\widehat{\gamma}$.

In case the diffraction is a pure point measure, a solution to the inverse problem has recently been presented by Lenz and Moody. In a comprehensive study \cite{lenzmoo}, which forms the basis for the current paper, they introduce a new kind of generalised point process called a spatial stationary process and provide a framework for constructing all such (real) processes that have a given diffraction. Several open questions remain, however, and we attempt to shed some light on some of these here. The central question considers exactly {\em what} it is that one may construct via this framework; that is, exactly which kinds of things may diffract?

In the first section, we outline the definitions and results from \cite{lenzmoo} that are required to apply the construction method and then present some simple examples of homometric structures. We find not only measures, but also ``non-measures'' (within the class of tempered distributions) that have lattice diffraction $\delta_\Z$. Section 2 contains a couple of results about diffraction solution classes; in particular we show that for a pure point diffraction measure, having a finite supporting set is both necessary and sufficient for its diffraction solution class to contain only measure-valued processes. Further, we find that any pure point diffraction measure with infinite supporting set must have uncountably many non-measures in its diffraction solution class. 

In the final section, we restrict from the general setting of locally compact Abelian groups to $\R^d$, and then suggest a possible change to the definitions of \cite{lenzmoo} that would allow some further classes of objects to fit into the framework. This is motivated by some of our examples, along with the observation that, although we can construct spatial stationary processes based on tempered distributions, there has not to date been a consideration of what kind of tempered distributions may have a well-defined autocorrelation.

\section{The construction process and some simple examples}

The approach of \cite{lenzmoo} is centred on a kind of generalised point process termed a spatial stationary process ({\em ssp}). There are essentially two key parts to the definition of an ssp, namely a dynamical system $(X,\mu,T)$, where the elements of the space $X$ are the objects which we might usually consider to diffract and the action of $T$ on $X$ can be thought of as translation, and a map $N$, which specifies the elements of the space $X$ via their `action' on elements of the space $C_c(\G)$. This description suggests that the elements of $X$ may be thought of as measures, which is often but not always the case. In any case, a spatial stationary process is implicitly real: there is a conjugacy condition on the map $N$ that ensures that the elements of $X$ are real-valued.

\begin{definition}[Definition 2.3, \cite{lenzmoo}]\label{def ssp}
Let $\G$ be a locally compact Abelian group, $(X,\mu)$ a probability space, $T$ a $\mu$-invariant action of $\G$ on $X$ that extends to a strongly continuous unitary representation of $\G$ on $L^2(X,\mu)$ via $T_tf(x) = f((-t)x)$, and $N: C_c(\G)\to L^2(X,\mu)$ a linear $\G$-map such that for all $f\in C_c(\G)$, $N(\overline{f}) = \overline{N(f)}$. Then the quadruple $\mathcal{N}:= (N,X,\mu,T)$ is called a {\em spatial stationary process} on $\G$. A spatial stationary process is called {\em ergodic} if the eigenspace of $T$ for the eigenvalue $0$ is one dimensional and {\em full} if the algebra generated by functions of the form $\psi\circ N(f)$, $\psi\in C_c(\C)$, $f\in C_c(\G)$, is dense in $L^2(X,\mu)$. 
\end{definition}

Before giving some examples, we need to define the diffraction of a stationary spatial process. The following is from \cite[Section 3.1]{lenzmoo}.

\begin{definition}\label{def auto-mom}
A spatial stationary process $\mathcal{N} = (N,X,\mu,T)$  on $\G$ is said to possess an 
{\em autocorrelation} if it has a second moment, that is, if and only if there exists a measure
$\gamma$ on $\G$ such that $\gamma(f\ast \widetilde{g}) = \<N(f),N(g)\>$ for all $f,g\in C_c(\G)$. For such a process $\mathcal{N}$, the {\em diffraction} of $\mathcal{N}$ is the Fourier transform of $\gamma$, the measure $\widehat{\gamma}$.
\end{definition}

If an autocorrelation exists, then by definition it must be real (as the map $N$ is real), positive definite and translation bounded. Accordingly, the diffraction is necessarily a positive, centrally symmetric measure. Note that as long as the group $\G$ has a countable basis of topology, the above definition of the autocorrelation of the ssp $\mathcal{N}$ coincides with usual notion of autocorrelation for a stationary point process (see \cite[Section 3.3]{lenzmoo}). 

Lenz and Moody further show that fullness of a process implies that pure 
pointedness of its dynamical and diffraction spectrum are equivalent \cite[Theorem 6.1]{lenzmoo}; see also \cite{baalenz}. Hence, 

\begin{definition}
A full spatial stationary process that possesses an autocorrelation is called {\em pure point} if its diffraction measure is a pure point measure. 
\end{definition}

We'll now take a quick look at how some familiar objects fit into this framework. 

\begin{example}\label{egZ}
Consider the measure $\delta_{\Z}$ on $\R$. We build a dynamical system via translation of the measure: let $\omega_0 := \delta_{\Z}$ and for all $t\in\R$, define $\omega_t := \delta_t \ast \omega_0 = \delta_{\Z+t}$. The space of all these translated measures is $X:= \{\omega_t\;|\; t\in \T \}$, where we identify $[0,1)$ with the torus $\T$, and, formally, we have a group action of $\R$ on $X$ defined for $t\in\R, s\in\T$ by 
$T_t \omega_s = \delta_t \ast \omega_s = \omega_{s+t}$. 
Obviously $X$ may also be identified with $\T$ (convolution being the group operation in $X$ that corresponds to addition in $\T$), and we define the measure $\mu$ on $X$ to be normalised Haar (Lebesgue) measure. 

We then get a (strongly continuous, unitary) representation of $\R$ on $L^2(X,\mu)$, defined for $t\in\R, s\in\T$ by 
\[ T_t f (w_s) := f(w_{s-t}) = f(T_{-t}w_s)\,,\]
although, identifying $X$ with $\T$ as above, we'll mostly write $f(t)$ for $f(\omega_t)$, and then 
$T_tf(s) = f(s-t)$. 

So, we see that $(X,\mu, T)$ is a pure point, ergodic dynamical system, with dynamical spectrum $\Z$ and generalised eigenfunction $\chi_k$ for each $k\in\Z$. 
(Here, and frequently in what follows, we use $\chi_s(t) = e^{2\pi \I st}$ for $s,t\in\R$.)

Now it is left to define the map $N : C_c(\R) \to L^2(X,\mu)$. The autocorrelation of our ssp must be $\gamma = \delta_{\Z}$ as this is the autocorrelation, in the usual sense (see for example  
\cite[chapter 9]{book}), of the measure $\delta_{\Z+t}$ for all $t$. 
Then for $f,g\in C_c(\R)$, one has
\[ \gamma(f\ast\widetilde{g}) = \check{\delta}_{\Z}(\widehat{f\ast\widetilde{g}}) 
			      = \delta_{\Z}(\widehat{f}\,\,\widehat{\widetilde{g}}) 
			      = \delta_{\Z}(\widehat{f}\,\,\overline{\widehat{g}}) \,,\]
where we have used the identity $\widehat{\omega}(f) = \omega(\widehat{f})$ in order to convert the convolution of functions to multiplication, then applied the Poisson summation formula 
$\widehat{\delta}_{Z} = \delta_{Z}$ and the identity $\overline{\widehat{f}} = \widehat{\widetilde{f}}$.
Then one simply needs to observe that
\[ \sum_{k\in\Z} \widehat{f}(k)\overline{\widehat{g}}(k) 
		= \int_{0}^1 \sum_{k,\ell\in\Z} \widehat{f}(k)\chi_k(t)\overline{\widehat{g}(\ell)\chi_\ell(t)} dt \,,\]
to see that to get $\gamma(f\ast\widetilde{g}) = \<N(f),N(g)\>$ we must define 
\[ N(f)(t) := \sum_{k\in\Z} \widehat{f}(k)\chi_k(t)\,,\;\; t\in\T\,.\]
Note that here (as in \cite{lenzmoo}) we use the convention of linearity in the {\em first} variable of the inner product.

Now writing
\[ N(f)(t) =  \sum_{k\in\Z} \widehat{T_{-t}f}(k) = \delta_{\Z}(\widehat{T_{-t}f}) 
						= \delta_{\Z}(T_{-t}f) = T_t\delta_{\Z}(f)\,,\]
using the Poisson summation formula again, one sees how the value of $N(f)$ at $t$ gives exactly the action of the `$t$-th element' of $X$ on the function $f$. That is, $N(f)(t) = \omega_t(f)$.
\end{example}

\begin{example}\label{eg 3per}
Let $\alpha\in\C$ and consider the measure 
$\omega_0 := \tfrac{1}{3} ( \chi_0 + \alpha\chi_1 + \overline{\alpha}\chi_2 ) \delta_{\frac{1}{3}\Z}$ on $\R$. 
Note that the choice of the coefficient $\overline{\alpha}$ for $\chi_2$ ensures that $\omega_0$ is real. 
As above, we can form the space $X:= \{\omega_t\;|\; t\in \T \}$ of all translations of $\omega_0$, and the same definitions of $\mu$ and $T$ give a dynamical system $(X,\mu,T)$. Then $(N,X,\mu,T)$ is a spatial stationary process, where for $f\in C_c(\R)$,
\[ N(f) = \sum_{k\in3\Z} \widehat{f}(k)\chi_k 
							+ \overline{\alpha}\!\!\!\sum_{k\in3\Z+1} \widehat{f}(k)\chi_k 
							+ \alpha\!\!\!\sum_{k\in3\Z+2} \widehat{f}(k)\chi_k\,.\]
The autocorrelation of $\omega_t$ is 
$\gamma:=\tfrac{1}{3} \left( \chi_0 + |\alpha|^2 (\chi_1 + \chi_2 )\right) \delta_{\frac{1}{3} \Z}$ for all $t$, and one can easily check that indeed, $\<N(f),N(g)\> = \gamma(f\ast\widetilde{g})$, and that 
$N(f)(t) = \omega_t(f)$ for $f,g\in C_c(\R)$, $t\in\T$.
\end{example}

Of course, things are not always so neat as in these very simple examples. In fact, we shall see that it's quite easy to construct some not-quite-so-neat examples. 
To this end, we'll introduce the construction method of Lenz and Moody via a bit more framework. The following definitions and results are from Sections 7 and 8 of \cite{lenzmoo}. Note that for a group $\G$,
we use $\widehat{\G}$ to signify its dual group. 

\begin{definition}\label{def S,ep}
Let $\widehat{\gamma}$ be a pure point diffraction measure, that is, a positive, centrally symmetric pure point measure on a group $\widehat{G}$ whose Fourier transform is also a measure.
Then the {\em supporting set} of $\widehat{\gamma}$ is
\[ \mathcal{S} = \mathcal{S}(\widehat{\gamma}) := \{ k\in \G \,|\, \widehat{\gamma}(k) > 0\}\,.\]
Let $\mathcal{E}_d := \<\mathcal{S}\>$, the group generated by the elements of $\mathcal{S}$ and given the discrete topology, and let $Y:= \widehat{\mathcal{E}_d}$ be its dual. 
As $\mathcal{E}_d$ is discrete, the group $Y$ is compact, so we may define a normalised Haar measure, $\lambda_Y$, on $Y$.
\end{definition}

\begin{definition}
The {\em relator group} of a pure point diffraction measure $\widehat{\gamma}$ is the group $\mathcal{Z}$ defined as the set of all equivalence classes of tuples of elements of 
its supporting set  $\mathcal{S}$ whose components sum to $0$, that is, the set $Z/\sim$, where 
\[ Z:= \{ (k_1,\ldots,k_n) \,|\, 
						n\in\Z,\, k_i \in\mathcal{S}\; {\rm for }\; 1\leq i\leq n,\, \sum_{i=1}^n k_i = 0 \}\,,\]
and under the equivalence relation $\sim$, two tuples are equivalent if one is a permutation of the other, or if one can be obtained from the other by inserting the element $0$ or
a pair of the form $\{k,-k\}$ for some $k\in\mathcal{S}$. The group operation here is concatenation. The relator group is a subgroup of the group $\mathcal{T}$ 
of all tuples under the equivalence relation, 
that is 
\[ \mathcal{Z} \subseteq \mathcal{T} 
			:= \{(k_1,\ldots,k_n) \,|\, n\in\Z,\, k_i \in\mathcal{S}\; {\rm for }\; 1\leq i\leq n \}/\sim \,.\] 
\end{definition}

\begin{definition}
Let $\widehat{\gamma}$ be a pure point diffraction measure with supporting set $\mathcal{S}$ and relator group $\mathcal{Z}$. A {\em phase form} on $\mathcal{Z}$ is an element of the dual group of 
$\mathcal{Z}$, that is, a group homomorphism $a^*: \mathcal{Z} \to \T$. An {\em elementary phase form} is a mapping $a:\mathcal{S} \to\T$ such that $a(0) = 1$ (if $0\in\mathcal{S}$)
and $a(-k):= \overline{a(k)}$ for all $k\in\mathcal{S}$.
\end{definition}

Note that elementary phase forms are in exact correspondence with elements of the dual group of $\mathcal{T}$. For $a:\mathcal{S}\to\T$, take $a_\mathcal{T}: \mathcal{T}\to\T$
defined (on a representative of an equivalence class) by 
$a_\mathcal{T}((k_1,\ldots,k_n)):= a(k_1)\cdot\ldots\cdot a(k_n)$. Then the conditions $a(0) = 1$ and 
$a(-k)= \overline{a(k)}$ ensure that $a_\mathcal{T}$ gives the same thing for any representative of an equivalence class and is a group homomorphism. 
Conversely, for  $a_\mathcal{T}\in \widehat{\mathcal{T}}$, defining $a(k) = a_\mathcal{T}((k))$ for $k\in\mathcal{S}$ gives a map $a:\mathcal{S}\to\T$ with the required properties. 
Viewing elementary phase forms in this way allows us to recognise
each phase form $a^*: \mathcal{Z} \to \T$ as a restriction of an elementary phase form 
$a: \mathcal{T} \to \T$.

Lenz and Moody show that two elementary phase forms restrict to the same phase form if and only if their ratio is in $Y$, that is, if and only if it defines a character on $\mathcal{E}_d$.

\begin{theorem}\cite[Proposition 8.1, Theorem 9.1]{lenzmoo}\label{thm diffphase}
Each ergodic pure point stationary process $\mathcal{N}$ on a locally compact Abelian group $\G$ gives rise to a unique pair $(\widehat{\gamma},a^*)$, where $\widehat{\gamma}$ is a pure point measure (the diffraction of $\mathcal{N}$) and $a^*$ is a phase form on the relator group of $\widehat{\gamma}$. The pair $(\widehat{\gamma},a^*)$ completely characterises the process $\mathcal{N}$ (up to isomorphism).
\end{theorem}

So given a diffraction measure $\widehat{\gamma}$ and a phase form (and an lca group), we should be able to construct a process that has diffraction $\widehat{\gamma}$. Actually, one uses elementary phase forms to construct a process (and elementary phase forms that restrict to the same phase forms give isomorphic processes). Note that we use $\chi_k(t)$ to represent the action of the dual group element $k$ on the group element $t$ (so it directly generalises our earlier definition).

\begin{proposition}\cite[Proposition 10.1]{lenzmoo}\label{thm construct}
Let $\widehat{\gamma}$ be a positive, symmetric (backward transformable) measure on an lca group $\widehat{\G}$, with supporting set $\mathcal{S}$. Let $a^*$ be a phase form
and $a$ be any elementary phase form that restricts to $a^*$. Define, for $f\in C_c(\G)$, 
\[ N_a(f) := \sum_{k\in\mathcal{S}} \widehat{f}(k) a(k) \widehat{\gamma}(k)^{\frac{1}{2}} \chi_k \,.\]
Then for $Y$, $\lambda_Y$ as in Definition \ref{def S,ep}, there is a natural ergodic action of $\G$ on $Y$, where for $t\in\G$ and $y\in Y$ the element $(t.y)$ of $Y$ is defined by $(t.y)(k):= \chi_k(t)y(k)$ for all $k\in\mathcal{E}_d$. Let $T$ be the corresponding representation of $\G$ on $L^2(Y,\lambda_Y)$ 
($T_tf(y):= f((-t).y)$).
Then $\mathcal{N}_a := (N_a,Y,\lambda_a,T)$ is an ergodic, full spatial stationary process on $\widehat{\G}$ with diffraction $\widehat{\gamma}$.
\end{proposition}

As should become evident through the following examples, this construction process (or, more specifically, the construction of the elements of the space $Y$ via
this process) may be viewed, in a quite natural way, as a reverse traversal of the Wiener diagram via the lower path (see \cite{teba} for more details).

\begin{example}
We begin with diffraction measure $\widehat{\gamma} := \delta_{\Z}$ (and underlying group $\R$) and create a process $\mathcal{N}_{a} = (N_{a},Y,\lambda_Y,T)$. Now here, 
\[ \mathcal{S}:= \{ k\in\Z : \widehat{\gamma}(k) > 0 \} = \Z = \left<\mathcal{S}\right> = \mathcal{E}_d\,, \]
so that $Y = \widehat{ \mathcal{E}_d} = \T$. For each elementary phase form $a:\Z\to \T$, we have a process with 
\begin{equation}
N_{a}(f) = \sum_{k\in\mathcal{S}} \widehat{f}(k) a(k) \widehat{\gamma}(k)^{\frac{1}{2}} \chi_k 
																	= \sum_{k\in\Z} \widehat{f}(k) a(k) \chi_k \in L^2(\T,\lambda_{\T}) \, \label{eq N for Z}
\end{equation}
for $f\in C_c(\R)$.

Clearly, choosing $a(k) = 1$ for all $k$ gives exactly the map $N$ we had above in Example \ref{egZ}, and thus returns the measure $\delta_{\Z}$ as the `zero' element of the space $Y$. More precisely, the map $N_a$, for this choice of $a$, tells us to interpret the elements of the space $Y$ as the measures 
$\{\delta_{\Z+t}\;|\; t\in \T\}$.

To create a slightly different process with diffraction $\delta_{\Z}$, choose $\alpha\in\T$ and then define $a:\Z\to \T$ via 
\[ a(k) = \left\{
	\begin{array}{ll}
	    1 	\,,		&  k\in 3\Z \\
	    \alpha \,, 		&  k\in 3\Z+1 \\
	    \overline{\alpha}\,,&  k\in 3\Z + 2 \,.
	\end{array}        \right.               \]
This is an elementary phase form, and substituting into (\ref{eq N for Z}), we derive a three-periodic measure process exactly as described in Example \ref{eg 3per}. 
One can then see easily enough how to choose an elementary phase form to construct 
an $n$-periodic measure with diffraction $\delta_{\Z}$. Similar examples were considered in \cite[Section 13]{lenzmoo}, although with the diffraction measure acting on a compact group rather than on $\R$.
\end{example}

Now to something non-periodic that nonetheless has diffraction $\delta_{\Z}$. This example first appeared in \cite{teba}.
 
\begin{example}\label{eg non-meas}
As above, we begin with diffraction measure $\widehat{\gamma} = \delta_{\Z}$ (and underlying group $\R$) and create a process 
$\mathcal{N}_{a} = (N_{a},Y,\lambda_Y,T)$,
where $Y=\T$ and for $f\in C_c(\R)$, $N_{a}(f) = \sum_{k\in\Z} \widehat{f}(k) a(k) \chi_k \in L^2(Y,\lambda_Y)$.

Define 
\[ L 	\;:=\;  2\Z \mathbin{\dot{\cup}} 
									\mathbin{\dot{\bigcup_{n\geq 1}}} (2.4^n \Z + (4^n - 1) ) 
									\mathbin{\dot{\cup}} (2.4^n \Z + (1-4^n))\,. \]
This is an aperiodic, symmetric set \cite{baamoo04, teba}, so that letting $a(k) = 1$ for $k\in L$ and $a(k) = -1$ for $k\notin L$ defines a valid elementary phase form $a:\Z\to\T$. 
Then
\begin{align*}
N_a(f)(t) 	&= \sum_{k\in\Z} \widehat{f}(k) a(k) \chi_k(t) \\
		&= \sum_{k\in L} \widehat{T_{-t}f}(k) - \sum_{k\in \Z\setminus L} \widehat{T_{-t}f}(k) \\
		&= (2 \delta_{L} - \delta_{\Z})(\widehat{T_{-t}f}) \\
		&= (2 \widehat{\delta}_{L} - \delta_{\Z}) (T_{-t}f) \\
		&= T_t\omega_0(f)\,,
\end{align*}
where $\omega_0:= 2 \widehat{\delta}_{L} - \delta_{\Z}$. Writing
\begin{align*}
\delta_{L}
	   &= \delta_{2\Z} + \sum_{n\geq 1} \delta_{2.4^n \Z} \ast (\delta_{4^n - 1} + \delta_{1-4^n})\,,
\end{align*}
we calculate a formal expression for $\widehat{\delta}_{L}$ by applying the Poisson summation formula to each term of the sum. That is, 
\begin{align*}
\widehat{\delta}_{L} 	&= \tfrac{1}{2}\delta_{\frac{\Z}{2}} + \sum_{n\geq 1} \frac{\overline{\chi}_{4^n - 1} 														+ \overline{\chi}_{1-4^n}}{2.4^n} \delta_{\frac{\Z}{2.4^n}} \\
											&= \tfrac{1}{2}\delta_{\frac{\Z}{2}} 
														+ \sum_{n\geq 1} \frac{\cos(2\pi(4^n-1)(\cdot))}{4^n}\delta_{\frac{\Z}{2.4^n}}\,.
\end{align*}
One can quickly verify that $\widehat{\delta}_{L}$ is not a measure. However, as $\delta_L$ is a translation bounded measure, $\widehat{\delta}_{L}$ is a tempered distribution. 
Thus we have constructed a process based on a ``non-measure'', $\omega_0$, that has diffraction $\delta_{\Z}$.  

As above, the map $N$ specifies the elements of the space $Y$ as $\{T_t \omega_0\;|\;t \in \T\}$. In contrast to the above examples, however, this is not quite what we might have expected. That is, $Y=\T$ does not  correspond to the hull of translations of $\omega_0 = 2 \widehat{\delta}_{L} - \delta_{\Z}$. If we call this hull $X$ (that is, let $X:= \overline{\{T_t\omega_0\;|\; t\in\R\}}$), then \cite[Theorem 15.3]{lenzmoo} says that $\T$ must be a factor of $X$ (which is clear) and also that we can get an isomorphism between the process we `expect' (that is, with carrying space $X$) and the one we construct from the diffraction. 
\end{example}

So, we can see that there are both measures and non-measure tempered distributions that have diffraction $\delta_{\Z}$. The following gives a clue as to a physical interpretation of a tempered distribution with a diffraction.

\begin{example}
Let $\varepsilon >0$ and define
\[ \nu_{\varepsilon} 
				:= \sum_{n\geq 1} \frac{\cos(2\pi(4^n - 1)(\cdot))}{(4+\varepsilon)^n} \delta_{\frac{\Z}{2.4^n}}\,.\]
We'll show that $\nu_{\varepsilon}$ is a measure for $\varepsilon>0$ by verifying that  $|\nu_{\varepsilon}|(1_{[0,\frac{1}{4}]})$ is finite.
Parametrising the supporting set of $\nu_{\varepsilon}$ (as in \cite{baamoo04}) as 
\[ F:=  \{ \tfrac{m}{2^\ell}\; |\; (\ell=0, m\in\Z) \;{\rm or}\; (\ell\geq 1, m \;{\rm odd}) \}\,,\]
we have 
\[ F\cap [0,\tfrac{1}{4}] 
						= \{0\} \cup \bigcup_{\ell\geq 3}\bigcup_{m=0}^{2^{\ell-3}-1} \tfrac{2m+1}{2^\ell}\,, \]
where for a given $\ell$ and $m$, the number $\frac{2m+1}{2^\ell}$ is an element of the lattice $\frac{\Z}{2.4^n}$ for all (integer) $n\geq \frac{\ell-1}{2}$. 
Then, for $\varepsilon>0$, we have 
\begin{align*}
|\nu_{\varepsilon}|(1_{[0,\frac{1}{4}]})
	  &\leq \sum_{n=1}^{\infty} \tfrac{1}{(4+\varepsilon)^n} 
	  		+ \sum_{\ell=3}^{\infty}\sum_{m=0}^{2^{\ell-3}-1}\sum_{n=\frac{\ell-1}{2}}^{\infty} 
	  									\frac{|\cos(\tfrac{2\pi(4^n - 1) (2m+1)}{2^\ell})|}{(4+\varepsilon)^n}  \\
	  &\leq \tfrac{1}{3+\varepsilon} 
	  		+ \tfrac{4+\varepsilon}{3+\varepsilon}\tfrac{1}{4+\varepsilon}\tfrac{\sqrt{4+\varepsilon}}
	  							{\sqrt{4+\varepsilon}-2} < \infty \,,
\end{align*}
as required, where the estimate is gained by bounding $|\cos{x}|$ by 1 and summing the geometric series. 

Now define $\rho_{\varepsilon}:= \frac{1}{2}\delta_{\frac{\Z}{2}} + \nu_{\varepsilon}$. Then $\rho_{\varepsilon}$ is a measure for all $\epsilon>0$, and as $\varepsilon \to 0$,
the measures $\rho_{\varepsilon}$ converge to the tempered distribution $\widehat{\delta}_L$ in the weak-$\ast$ topology on the space $\mathfrak{S}(\R)^{\prime}$ of tempered distributions. 

Consider now the measures $\omega^{\varepsilon}= 2\rho_{\varepsilon} - \delta_{\Z}$ for $\varepsilon > 0$ (so that $\omega^{\varepsilon} \to \omega_0:=2 \widehat{\delta}_{L} - \delta_{\Z}$ in the weak-$\ast$ topology on $\mathfrak{S}(\R)^{\prime}$ as $\varepsilon\to 0$). 
A short calculation shows that the diffraction of $\omega^{\varepsilon}$ is
\begin{align*}
\widehat{\gamma}_{\varepsilon} 	
			&:=  \delta_{2\Z} + \sum_{n\geq 1}\left(2(\tfrac{4}{4+\varepsilon})^n - 1\right)^{2} \delta_{L_n} 
																+ \delta_{\Z\setminus L}\,
\end{align*}
where $L_n := (2.4^n \Z + (4^n - 1) ) \cup (2.4^n \Z + (1-4^n))$ for all $n\geq 1$. For each 
$\varepsilon > 0$, $\widehat{\gamma}_{\varepsilon}$ is a positive, symmetric measure, and $\{\widehat{\gamma}_{\varepsilon}\}_{\varepsilon>0}$ converges (weak-$\ast$) to $\delta_{\Z}$, the diffraction of $\omega_0$, as $\varepsilon\to 0$.

Now fix an $\varepsilon>0$. Via the reconstruction method of Lenz and Moody, we can use $\widehat{\gamma}_{\varepsilon}$ to reconstruct the measure $\omega^{\varepsilon}$, 
as well as other objects with the same diffraction. For $f\in C_c(\R)$, define 
\begin{align*}
N^{\varepsilon}_a(f)(0)	
			&= \sum_{k\in\mathcal{S}} \widehat{f}(k) a(k) \widehat{\gamma}_{\varepsilon}(k)^{\frac{1}{2}} \\
			&= \sum_{k\in 2\Z} \widehat{f}(k)a(k) 
								+ \sum_{n\geq 1} \left(2(\tfrac{4}{4+\varepsilon})^n - 1\right)
												\sum_{k\in L_n} \widehat{f}(k)a(k) + \sum_{k\in\Z \setminus L}\widehat{f}(k)a(k)\,.
\end{align*}
Observe that
$\widehat{\nu}_{\varepsilon} = \sum_{n\geq 1} (\tfrac{4}{4+\varepsilon})^n \delta_{L_n}$.
Then choosing $a:\mathcal{S}\to \T$ with $a(k)= 1$ for $k\in L$ and $a(k)= -1$ otherwise gives 
\begin{align*}
N^{\varepsilon}_a(f)(0) 
	&= \delta_{2\Z}(\widehat{f}) 
								+ 2\widehat{\nu}_{\varepsilon}(\widehat{f}) - \delta_{L\setminus2\Z}(\widehat{f}) 
								- \delta_{\Z\setminus L}(\widehat{f})\\
	&= 2(\delta_{2\Z}(\widehat{f})+\widehat{\nu}_{\varepsilon}(\widehat{f}))-\delta_{\Z}(\widehat{f})\\
	&= \delta_{\frac{\Z}{2}}(f) + 2\nu_{\varepsilon}(f) - \delta_{\Z}(f)\\
	&= \omega^{\varepsilon}(f)\,,
\end{align*}
as expected.
Choosing the elementary phase form $a: \Z\to\T$ with $a(k) = 1$ for all $k\in\Z$ gives
\begin{align*}
N^{\varepsilon}_a(f)(0) 
	&= \delta_{2\Z}(\widehat{f}) + 2\widehat{\nu}_{\varepsilon}(\widehat{f}) 
							- \delta_{L\setminus2\Z}(\widehat{f}) + \delta_{\Z\setminus L}(\widehat{f})\\
	&= 2(\delta_{2\Z}(\widehat{f})+\widehat{\nu}_{\varepsilon}(\widehat{f})-\delta_{L}(\widehat{f}))
							+\delta_{\Z}(\widehat{f})\\
	&= \delta_{\frac{\Z}{2}}(f) + 2\nu_{\varepsilon}(f) + \delta_{\Z}(f)-2\widehat{\delta}_{L}(f)\\
	&= (\omega^{\varepsilon} +2(\delta_{\Z} - \widehat{\delta}_{L}))(f) \\
	&=: \sigma_{\varepsilon}(f) \,,
\end{align*}
where $\sigma_{\varepsilon}$ is not a measure, as $\widehat{\delta}_{L}$ is not. So (of course!) the homometry class of $\omega^{\varepsilon}$ contains non-measures as well as measures. 
And this time, we have that as $\varepsilon \to 0$, we get a measure, namely $\delta_{\Z}$. (Note that for each $\varepsilon$ we can again construct a sequence of measures that weak-$\ast$
converges  to the tempered distribution $\sigma_{\varepsilon}$.)
 
So, we have a family of tempered distributions, $\sigma_{\varepsilon}$, that are not measures 
for $\varepsilon > 0$, but tend to the measure $\delta_{\Z}$ in the weak-$\ast$ topology as $\varepsilon$ goes to zero. 
For each $\varepsilon > 0$, the homometry class of $\sigma_{\varepsilon}$ contains an element $\omega^{\varepsilon}$ that is a measure, however as $\varepsilon$ goes to zero 
these $\omega^{\varepsilon}$ tend to a non-measure, $\omega_0$. 
\end{example}

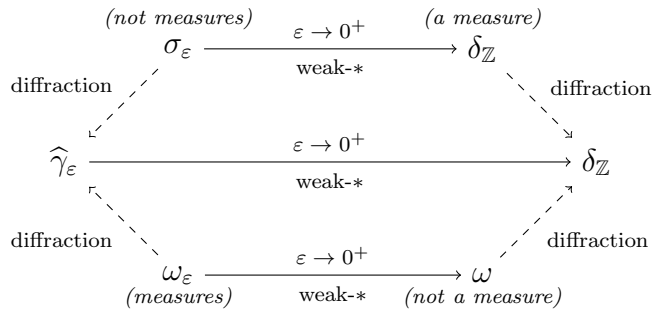
\begin{center}
\begin{tikzpicture}[node distance=4cm, auto]
\node (sigep) {$\sigma_{\varepsilon}$};
\node (rhoep) [below of=sigep, node distance=3cm] {$\omega_{\varepsilon}$};
\node (diffep) [left of=sigep, above of=rhoep, node distance=1.5cm] {$\widehat{\gamma}_{\varepsilon}$};
\node (rho0) [right of=rhoep] {$\omega$};
\node (deltaz) [right of = sigep] {$\delta_{\Z}$};
\node (diffdeltz) [below of=deltaz, right of=deltaz, node distance=1.5cm] {$\delta_{\Z}$};

\node (notmeasl) [above of=sigep, node distance=.4cm] {{\tiny\em (not measures)} };
\node (measr) [above of=deltaz, node distance=.4cm] {{\tiny\em (a measure)} };
\node (measl) [below of=rhoep, node distance=.3cm] {{\tiny\em (measures)} };
\node (notmeasr) [below of=rho0, node distance=.3cm] {{\tiny\em (not a measure)} };

\draw[->, dashed] (sigep) to node [swap] {{\tiny diffraction}} (diffep);
\draw[->, dashed] (rhoep) to node {{\tiny diffraction}} (diffep);

\draw[->, dashed] (deltaz) to node {{\tiny diffraction}} (diffdeltz);
\draw[->, dashed] (rho0) to node [swap] {{\tiny diffraction}} (diffdeltz);

\draw[->] (sigep) to node {{\tiny $\varepsilon\to 0^+$}} (deltaz);
\path (sigep) to node [swap] {{\tiny weak-$\ast$}} (deltaz);

\draw[->] (rhoep) to node {{\tiny $\varepsilon\to 0^+$}} (rho0);
\path (rhoep) to node [swap] {{\tiny weak-$\ast$}} (rho0);

\draw[->] (diffep) to node {{\tiny $\varepsilon\to 0^+$}} (diffdeltz);
\path (diffep) to node [swap] {{\tiny weak-$\ast$}} (diffdeltz);

\end{tikzpicture}
\end{center}

We'll come back to these considerations in Section \ref{sec Extend}, where we restrict to $\R^d$ to suggest a broadening of the definition of objects that may have a diffraction. 

\section{A couple of results}

In this section, we present a couple of results regarding diffraction solution classes. Example \ref{eg non-meas} shows that even a very simple periodic measure such as $\delta_{\Z}$ may contain a quite unexpected object in its diffraction solution class, but of course this example is not unique. In fact, there are uncountably many such examples. 

\begin{theorem}\label{thm infmany}
The diffraction solution class of $\delta_{\Z}$ contains uncountably many `non-measures'.
\end{theorem}
\begin{proof}
By Theorem \ref{thm construct}, a process with diffraction $\delta_{\Z}$ is based on an object $\omega^a_0$, where $a:\Z\to\T$ is an elementary phase form and for $f\in C_c(\R)$,
\[ \omega^a_0(f) := N_a(f)(0) = \sum_{k\in\Z} \widehat{f}(k) a(k) \,.\]
Let $\Lambda \subseteq\Z$ be an infinite, symmetric, aperiodic set that contains $0$ and define $a(k) = 1$ for $k\in \Lambda$, $a(k) = -1$ otherwise. Then $a$ defines an elementary phase form and for $f\in C_c(\R)$,
\[ \omega^a_0(f) 	= \sum_{k\in\Lambda} \widehat{f}(k) - \!\!\sum_{k\in\Z\setminus\Lambda}\widehat{f}(k) 
									= (2\delta_{\Lambda}-\delta_{\Z})(\widehat{f}) \,,\]
so that $\omega^a_0 = 2\widehat{\delta}_{\Lambda} - \delta_{\Z}$. As an aperiodic set, $\Lambda$ is not a finite union of lattices, so by a theorem of Cordoba \cite{cordo}, $\widehat{\delta}_{\Lambda}$ is not a measure. Then $\omega^a_0$ is not a measure. There are uncountably many such sets $\Lambda$ (that are not translations of each other), so there are uncountable many (non-isomorphic) processes based on a non-measure that have diffraction $\delta_{\Z}$.
\end{proof}

On the other hand, there are measures for which one may quite confidently describe the entire diffraction solution class in concrete terms.  

\begin{example}[The diffraction solution class of $\delta_0$.]
Let $\omega = \lambda$, Lebesgue measure, on $\R$. Then as Lebesgue measure is translation invariant, $\delta_t\ast \omega = \omega$ for all $t\in\R$, and we get the dynamical system $(X,\mu,T)$ with $X:=\{\omega\}$, $T_t\omega := \omega$ for all $t\in\R$ and $\mu(\{\omega\}) := 1$. 
(Identifying $X$ with $\{0\}$, we'll use $0$ as ``shorthand'' (a convenient name) for $\omega$.)

Define $N:C_c(\R) \to L^2(X,\mu)$ via $N(f)(0) := \lambda(f) = \widehat{f}(0)$ for all $f\in C_c(\R)$. Then $\mathcal{N}:=(N,X,\mu,T)$ is a spatial stationary process with diffraction measure 
$\widehat{\gamma} := \delta_0$. 

We have $\mathcal{S} = \{0\} = \widehat{\mathcal{S}} = X$, and thus we construct the measure $\rho_0$, where for $f\in C_c(\R)$, 
\[ \rho_0(f) := N_a(f)(0) = a(0)\widehat{f}(0)\,.\]
As the only legal choice for an elementary phase form $a:\mathcal{S}\to\T$ is given by $a(0) = 1$, we have
\[ \rho_0(f) = \widehat{f}(0) = \lambda(f) \,.\]
So we have $\rho_0 = \lambda$, that is, we've reconstructed our original process. By taking the negative of this process, that is, $-\mathcal{N}$, we gain another
process with the same diffraction, namely $-X = \{-\lambda\}$, and thus see that the {\em real} homometry class of Lebesgue measure contains only $\lambda$ and $-\lambda$. 
\end{example}

One of the open questions posed in \cite[Section 15]{lenzmoo} is whether conditions can be placed on a diffraction measure to ensure that it stems only from measure-valued processes. It turns out that we need to restrict quite a lot.

\begin{example}[Measures whose diffraction solution class contains only measures.]\label{eg finsupp}
A slight variation of the previous example gives a class of examples with diffraction measure of finite support. These turn out to be exactly the diffraction measures 
who have no non-measures in their diffraction solution class. So, suppose that $\widehat{\gamma}$ is a diffraction measure with finite supporting set on $\R$. As $\widehat{\gamma}$ is symmetric, we may write the supporting set as
\[ \mathcal{S} 	= \mathcal{S}(\widehat{\gamma}) 
								= \{ x_j,\, -x_j | j=1,2,\ldots , n \} \cup \left\{ 
																\begin{array}{cl}
																		    \{0\}\,, 	&  \widehat{\gamma}(0) \neq 0 \\
																		    \varnothing\,,&  \widehat{\gamma}(0) = 0 \,
																\end{array}        \right. \,.\]
Then, as $\widehat{\gamma}$ must also be real and positive, we can write 
\[ \widehat{\gamma} := b_0^2 \delta_0 + \sum_{j=1}^n b_j^2 (\delta_{x_j}+ \delta_{-x_j}) \,,\]
for some $\{b_j\}\subseteq\R$, using the convention that if $0\not\in\mathcal{S}$, $b_0 = 0$. 
For $f\in C_C(\R)$, we have
\begin{align*}
N_a(f)(0) 	&= \sum_{k\in\mathcal{S}} \widehat{f}(k) a(k) \widehat{\gamma}(k)^{\frac{1}{2}} \\	
		&= |b_0|(\widehat{f}(0)) + \sum_{j=1}^n |b_j|(a(x_j)\widehat{f}(x_j) + a(-x_j)\widehat{f}(-x_j))\\
		&= |b_0|(\widehat{f}(0)) + \sum_{j=1}^n |b_j|(a(x_j)\widehat{f}(x_j) 
							+ \overline{a(x_j)}\widehat{f}(-x_j))\\
		&= |b_0|\delta_0(\widehat{f}) + \sum_{j=1}^n |b_i|(a(x_j)\delta_{x_j} 
							+ \overline{a(x_j)}\delta_{-x_j})(\widehat{f})\\
		&= |b_0|\chi_0(f) + \sum_{j=1}^n |b_j|(a(x_j)\overline{\chi_{x_j}} + \overline{a(x_j)}\chi_{x_j})(f) \,,
\end{align*}
so that we've constructed a measure of the form $\rho_a = g_a\lambda$, with
\begin{align*}
g_a(t) 	&= |b_0|\chi_0(t) + \sum_{j=1}^n |b_j|(a(x_j)\overline{\chi_{x_j}}(t) 
											+ \overline{a(x_j)}\chi_{x_j}(t)) \\
				&= |b_0|  + 2\sum_{j=1}^n |b_j|\cos(2\pi(a_j - x_j t))\,,\; t\in\R\,,
\end{align*}
where for each $j$, $a(x_j) = e^{2\pi \I a_j}$, $a_j\in[0,1)$.
\end{example}

This construction (up to the second-last line above) can be generalised directly to a diffraction measure on any locally compact Abelian group $\G$, where (as defined earlier) we use $\chi_k(t)$ to mean the action of the dual group element $k$ on the group element $t$, and $\lambda$ to mean appropriately chosen Haar measure on the dual group $\widehat{\G}$.

\begin{theorem}\label{thm finsupp}
Let $\widehat{\gamma}$ be a pure point diffraction measure on a locally compact Abelian group $\G$. Then the diffraction solution class of $\widehat{\gamma}$ contains only measures if and only if the supporting set of $\widehat{\gamma}$ is finite. In such a case, each process in the homometry class is based on a measure that is absolutely continuous with respect to Haar measure on $\G$. 
\end{theorem}
\begin{proof}
The above example shows the forward implication and the second statement. Now suppose that 
$\widehat{\gamma}$ is a diffraction measure with infinite supporting set $\mathcal{S}$. Choose a symmetric subset of $\mathcal{S}$ (not necessarily proper) that we can match up with elements of $\Z$, 
call it $\mathcal{S}_{\Z}$ and write $\mathcal{S}_{\Z} = \{x_j | j\in\Z\}$. Then, using the set $L$ from Example \ref{eg non-meas}, define, for $x\in \mathcal{S}$,
\[ a(x) = \left\{
	\begin{array}{rl}
	    1 \,,			&  x=x_j\,,\, j\in L \\
	    -1 \,,		&  {\rm otherwise}\,.
	\end{array}        \right.               \]
The function $a:\mathcal{S}\to\T$ is an elementary phase form and, as in Example \ref{eg non-meas}, it gives rise to a non-measure valued process.
\end{proof}

We reformulate the above result in the form of the following two corollaries.

\begin{corollary}
Let $\G$ be a locally compact Abelian group and $\lambda$ Haar measure on $\G$. Suppose that $\gamma$ is a positive definite measure on $\G$ such that $\widehat{\gamma}$ is a pure point measure. If all solutions, $\omega$, to $\gamma = \omega \econv \widetilde{\omega}$ are measures, then $\gamma$ has the form 
$\gamma = g\lambda$, where 
\[ g = b_0^2 + \sum_{j=1}^n b_j^2 (\chi_{x_j} +\chi_{-x_j}) \,,\]
and where $x_j \in \G$, $b_j \in \R$ for all $i$.
\end{corollary}

\begin{corollary}
For a given pure point diffraction measure $\widehat{\gamma}$ on a group $\widehat{\G}$, the (real-valued) map on $C_c(\G)$ defined
by $f\mapsto \sum_{k\in \widehat{\G}} \widehat{f}(k)a(k)\widehat{\gamma}^{\frac{1}{2}}(k)$ is bounded (in some appropriate $K$-norm)
for every elementary phase form $a:\widehat{\G}\to\T$ if and only if $\widehat{\gamma}$ has finite supporting set.
\end{corollary}
\begin{proof}
For each elementary phase form $a$, the map $m_a: C_c(\G)\to \R$ defined by 
\[ f\mapsto N_a(f)(0) = \sum_{k\in \widehat{\G}} \widehat{f}(k)a(k)\widehat{\gamma}^{\frac{1}{2}}(k)\]
is linear, as $N_a: C_c(\G) \to L^2(X,\mu)$ is linear, so that $m_a \in C_c(\G)^*$ if and only if the map is bounded. Letting $a$ run through all possible elementary phase forms
gives us the homometry class. Then the proof of Theorem \ref{thm finsupp} adapts directly (if the supporting set is infinite, then we can choose a symmetric, aperiodic subset and get a non-measure).
\end{proof}

Note that in the proof of Theorem \ref{thm finsupp}, the set $L$ was used simply for convenience. Any infinite, symmetric, aperiodic set as in the proof of 
Theorem \ref{thm infmany} would have sufficed. Thus, combining the proofs of Theorems \ref{thm finsupp} and \ref{thm infmany} gives the following.

\begin{corollary}\label{cor uncount}
Let $\widehat{\gamma}$ be a pure point diffraction measure with infinite supporting set on a locally compact Abelian group $\G$. The diffraction solution class of $\widehat{\gamma}$ contains uncountably many non-measures.
\end{corollary}

Given that a measure representing the diffraction of a physical structure has an infinite supporting set, we see that the only measures with a ``simple'' diffraction solution class, as in Theorem \ref{thm finsupp}, are those that do not represent a physical structure. Corollary \ref{cor uncount} shows that a physical structure must share its diffraction with uncountably many non-measures.

\section{An alternative framework}\label{sec Extend}

Take the tempered distribution $\rho_0:= \widehat{\delta}_L$ on $\R$, as defined in Example \ref{eg non-meas}, and form a spatial stationary process $\mathcal{N}:=(N,X,\mu,T)$ with 
$X:= \overline{\{T_t\rho_0\;|\; t\in\R\}}$ equipped with some ergodic probablity measure $\mu$. Define the map $N$ on $C_c(\R)$ by
\[ N(f):= \sum_{k\in L} \widehat{f}(k)\chi_k \,,\; f\in C_c(\R)\,.\]
A short calculation via the inner product verifies that $\norm{N(f)}_2^2  = \widehat{\delta}_L(f\ast\widetilde{f})$ for all $f\in C_c(\R)$. However, this is not necessarily finite, 
as $f$ (and hence $f\ast\widetilde{f}$) is not necessarily a Schwartz function. That is, we don't necessarily have that $N(f)$ is in $L^2(X,\mu)$. 

What is it that makes $\widehat{\delta}_L$ unable to fit into the framework of spatial stationary processes, while its neighbour, the tempered distribution 
$\omega_0 = 2 \widehat{\delta}_{L} - \delta_{\Z}$ of Example \ref{eg non-meas}, has no problems? 
The latter has diffraction, and hence autocorrelation, $\delta_{\Z}$. Using this, we may calculate the autocorrelation of the former.

Firstly, recall that the autocorrelation $\gamma$ of a translation bounded measure $\mu$ is defined via volume averaged convolution:
\[ \gamma := \mu\econv\widetilde{\mu} 
					:= \lim_{R\to\infty} \frac{\mu|_R \ast\widetilde{\mu|_R}}{\lambda(B_R)}\,, \]
where $B_R$ is the ball of radius $R$, $\mu|_R$ is the restriction of the measure $\mu$ to $B_R$, the measure $\widetilde{\mu}$ is defined by 
$\widetilde{\mu}(f) := \overline{\mu(\widetilde{f})}$, where $\widetilde{f}(t):= \overline{f(-t)}$, and the limit is taken in the weak-$\ast$ topology in the space of measures. If this limit exists, the autocorrelation is itself a translation bounded measure \cite[Proposition 2.2]{hof95}. 

Returning to our consideration of $\delta_L$, we have
\[ \delta_{\Z} \;\;= \;\; \omega_0\econv \widetilde{\omega_0} 
								\;\; = \;\;4 \widehat{\delta}_{L}\econv \widehat{\delta}_{L} - 4 \widehat{\delta}_{L} \econv \delta_{\Z} + \delta_{\Z} \]
so that $\widehat{\delta}_{L}\econv \widehat{\delta}_{L} = \widehat{\delta}_{L} \econv \delta_{\Z}$. (Observe that as $\widehat{\delta}_L$ is real and symmetric,
$\widetilde{\widehat{\delta}_L} = \widehat{\delta}_L$.)
Applying the identity $\delta_{\frac{\Z}{p}}\econv \delta_{\Z} = \delta_{\frac{\Z}{p}}$ \cite[Ch 8]{book} to each term of the formal sum representation of $\widehat{\delta}_{L}$, one derives that $\widehat{\delta}_{L} \econv \delta_{\Z} = \widehat{\delta}_{L}$ and thus sees that $\widehat{\delta}_{L}$ has autocorrelation $\widehat{\delta}_{L}$ and diffraction ${\delta}_{L}$. 

So the ``problem'' here is that the diffraction of the spatial stationary process $\mathcal{N}$ is not backward transformable, that is, the autocorrelation is not a measure. The condition that the autocorrelation of a spatial stationary process be a measure (cf Definition \ref{def auto-mom}) ensures that the map $N$ does indeed map $C_c(\R)$ functions into $L^2(X,\mu)$ functions.

Is there however really an intrinsic problem with having a nice, positive, symmetric diffraction measure like $\delta_L$? What would be lost, if anything, in changing the `evaluation space' of a spatial stationary process from one of continuous functions of compact support to Schwartz functions? We then restrict ourselves to underlying group $\R^d$, so lose the generality of an lcag, but would seem to gain a lot.

What to allow then as an autocorrelation, and what class of objects to restrict to to ensure that we can always define the autocorrelation? Firstly, we expect an autocorrelation to be positive definite. The Bochner-Schwartz theorem (see, for example, \cite[Chapter II]{gelf} or \cite[Chapter IX]{reesim}) characterises positive definite distributions as those that are the Fourier transform 
of a tempered measure. This means that such a distribution is itself tempered, so we can aim to define autocorrelation for some class of tempered distributions. 

Motivated by the examples at the end of Section 1, one could think of taking the closure, in the weak-$\ast$ topology, of the set of tempered measures for which one may define a diffraction, that is, those tempered measures which possess a unique (natural) autocorrelation (see \cite{hof95}, \cite[Chapter 9]{book}). Perhaps this is possible, however, as a characterisation of the set of tempered measures that possess an autocorrelation is not yet available, we restrict to a set of tempered measures that we do know a little about, namely, translation bounded measures.

If we have a sequence $\{\mu_n\}$ of translation bounded measures such that the autocorrelation, $\gamma_n$, exists for each $\mu_n$, then, by \cite[Proposition 2.2]{hof95}, $\{\gamma_n\}$ itself is a sequence of translation bounded measures. If the weak-$\ast$ limit of $\{\mu_n\}$ in $\mathfrak{S}(\R^d)^{\prime}$ is the tempered distribution $T$, then surely a sensible definition of the autocorrelation of $T$ would simply be the weak-$\ast$ limit of the autocorrelations $\{\gamma_n\}$. For this to be indeed sensible, it must be consistent, that is, if the limit of $\{\mu_n\}$ is itself a measure, then we must get for the ``tempered distribution autocorrelation'' what we would get via the usual convolution calculation. 

\begin{conjecture}
Let $\{\mu_n\}\subseteq \mathfrak{S}(\R^d)^{\prime}$ be a sequence of uniformly translation bounded measures, each possessing a natural autocorrelation, such that 
the weak-$\ast$ limit, $\mu$, of $\{\mu_n\}$ in $\mathfrak{S}(\R^d)$ is also a translation bounded measure. Then for all $f\in \mathfrak{S}(\R^d)$,
\[\mu\econv\widetilde{\mu}\,(f) = \lim_{n\to\infty} \mu_n\!\econv\widetilde{\mu_n} \,(f) \,. \]
\end{conjecture}

If this were true, we would have the following.

\begin{definition}
Let $\{\mu_n\}\subseteq \mathfrak{S}(\R^d)^{\prime}$ be a sequence of uniformly translation bounded measures, each possessing a natural autocorrelation, that converges weak-$\ast$ to the tempered distribution $T$. Then the {\em autocorrelation} of $T$ is the tempered distribution $T \econv \widetilde{T}$ defined by
\[ T \econv \widetilde{T}\,(f) := \lim_{n\to\infty}\,\mu_n\!\econv\widetilde{\mu_n}\, (f) \]
for all $f\in \mathfrak{S}(\R^d)$.
\end{definition}

To adapt the framework of \cite{lenzmoo} to include this class of objects, we would begin by making a small alteration to Definitions \ref{def ssp} and \ref{def auto-mom}.

\begin{definition}\label{def ssp new}
Let $(X,\mu)$ be a probability space, $T$ a $\mu$-invariant action of $\R^d$ on $X$ that extends to a strongly continuous unitary representation 
of $\G$ on $L^2(X,\mu)$ via $T_tf(x) = f((-t)x)$, and $N: \mathfrak{S}(\R^d)\to L^2(X,\mu)$ a linear $\R^d$-map such that for all $f\in \mathfrak{S}(\R^d)$, $N(\overline{f}) = \overline{N(f)}$. 
Then the quadruple $\mathcal{N}:= (N,X,\mu,T)$ is called a {\em well-tempered spatial stationary process} on $\R^d$. A spatial stationary process is called {\em ergodic}
if the eigenspace of $T$ for the eigenvalue $0$ is one dimensional and {\em full} if the algebra generated by functions of the form $\psi\circ N(f)$, $\psi\in C_c(\C)$, $f\in \mathfrak{S}(\R^d)$, is dense in $L^2(X,\mu)$. 
\end{definition}

\begin{definition}\label{def auto-mom new}
A well-tempered spatial stationary process $\mathcal{N} = (N,X,\mu,T)$ on $\R^d$ is said to possess an {\em autocorrelation} if and only if there exists a tempered
distribution $\gamma$ on $\G$ such that $\gamma(f\ast \widetilde{g}) = \<N(f),N(g)\>$ for all $f,g\in \mathfrak{S}(\R^d)$. For such a process $\mathcal{N}$, the 
{\em diffraction} of $\mathcal{N}$ is the Fourier transform of $\gamma$, the measure $\widehat{\gamma}$.
\end{definition}

Note that $\widehat{\gamma}$ is a positive tempered measure by the Bochner-Schwartz Theorem, as by the definition, the autocorrelation $\gamma$ is positive definite.

Definitions are all very well and good, but they will only be useful if we can show that the rest of the framework of \cite{lenzmoo} will adapt on through. However, an inspection of the relevant results suggests that this should not pose too much of a problem. The results of \cite[Section 3.1]{lenzmoo} are true with $\mathfrak{S}(\R^n)$ in place of $C_c(\R^n)$ due to the fact that $\mathfrak{S}(\R^n)$ is dense in $C_c(\R^n)$ and in $L^2(\R^n)$. To show that the diffraction to dynamics map of \cite[Section 3.2]{lenzmoo} exists, one requires that $\mathfrak{S}(\R^n)$ be dense in $L^2(\R^n,\widehat{\gamma})$ for a given diffraction measure $\widehat{\gamma}$. This is certainly true if $\widehat{\gamma}$ is a translation bounded measure, which suggests that we may need to restrict to diffraction measures that are translation bounded. From a physical perspective, this seems quite reasonable. Doing this would give us the remainder of the results in \cite[Section 3.2]{lenzmoo}, which in turn are the basis for 
Theorems \ref{thm diffphase} and \ref{thm construct}, that is, the characterisation and construction results that we would like to use. 

It certainly seems feasible, then, to widen the framework of spatial stationary processes to include those with tempered distributions as autocorrelations. This would allow the consideration of a tempered distribution such as $\widehat{\delta}_L$ within the framework, and would presumably also open the door to the construction of even more interesting and unexpected objects. Although we have gained a little insight into what kinds of objects can display a pure point diffraction, and just how large a diffraction solution class can really be, we are still far from a complete answer to Bombieri and Taylor's question \cite{bombtay} of ``which distributions of matter diffract?''.

\end{document}